\documentclass[aps,preprint]{revtex4}%
\usepackage{amssymb}
\usepackage{amsmath}
\usepackage[dvips]{graphicx}
\usepackage{amsfonts}%
\providecommand{\U}[1]{\protect\rule{.1in}{.1in}}
\newtheorem{theorem}{Theorem}

\newtheorem{condition}[theorem]{Condition}

\newtheorem{definition}[theorem]{Definition}
\newtheorem{example}[theorem]{Example}

\newtheorem{lemma}[theorem]{Lemma}

\newtheorem{proposition}[theorem]{Proposition}

\newenvironment{proof}[1][Proof]{\textbf{#1.} }{\ \rule{0.5em}{0.5em}}
\begin{document}
\title{On ergodicity for multi-dimensional harmonic oscillator systems with
Nose-Hoover type thermostat }
\author{Ikuo Fukuda$^{1}$, Kei Moritsugu$^{2}$, and Yoshifumi Fukunishi$^{3}$}
\affiliation{$^{1}${\normalsize Institute for Protein Research, Osaka University, 3-2
Yamadaoka, Suita, Osaka 565-0871, Japan and Graduate School of Simulation
Studies, University of Hyogo, Kobe 650-0047, Japan}}
\affiliation{{\normalsize \noindent}$^{2}${\normalsize Graduate School of Medical Life
Science, Yokohama City University, Yokohama 230-0045, Japan}}
\affiliation{$^{3}${\normalsize Cellular and Molecular Biotechnology Research Institute,
AIST}}

\begin{abstract}
A simple proof and detailed analysis on the non-ergodicity for
multidimensional harmonic oscillator systems with Nose-Hoover type thermostat
are given. The origin of the nonergodicity is symmetries in the
multidimensional\ target physical system, and is differ from that in the
Nose-Hoover thermostat with\ the 1-dimensional harmonic oscillator. A new
simple deterministic method to recover the ergodicity is also presented. An
individual thermostat variable is attached to each degree of freedom, and all
these variables act on a friction coefficient for each degree of freedom. This
action is linear\ and controlled by a Nos\'{e} mass matrix $\mathbf{Q}$, which
is a matrix analogue of the scalar Nos\'{e}'s mass. Matrix $\mathbf{Q}$\ can
break the symmetry and contribute to attain the ergodicity.

\end{abstract}
\date{August 16, 2020}
\maketitle

\section{Introduction\label{Introduction}}

The Nos\'{e}-Hoover (NH)~\cite{N-H eq(1),N-H eq(2)} equation has been utilized
as basic equations of motion (EOM) in molecular dynamics (MD), which is now an
important tool to perform a realistic simulation of a physical
system~\cite{Hoover book1,AT,Schlick}. The NH equation is an ordinary
differential equation (ODE), based on the Newtonian EOM described by physical
coordinates $x\in\mathbb{R}^{n}$\ and momenta $p\in\mathbb{R}^{n}$. It is
obtained by adding a friction force that is $-(\zeta/Q)p$\ to the Newtonian
EOM and by adding a EOM\ for $\zeta$, where $Q$\ is a real parameter often
called as the Nos\'{e}'s mass. The friction coefficient-like quantity
$\zeta\in\mathbb{R}$ is thus a dynamical variable, and it is introduced to
control the temperature of the target physical system described by
$(x,p)$\ and maintain the value around a desired value $T_{\mathrm{ex}}$. It
is shown that the physical system obeys the Boltzmann-Gibbs (BG), or
canonical, distribution at temperature $T_{\mathrm{ex}}$\ if the total system
described by $(x,p,\zeta)$\ satisfies the ergodic condition.\ 

1-dimensional harmonic oscillator (1HO) has been investigated, for
theoretically studying the NH equation,\ as the most simple model system that
describes near a physical equilibrium. In a viewpoint of dynamical system
study, the NH equation with a 1HO has first been studied
numerically~\cite{PHV} and revealed to include both regular and chaotic
motions, implying that NH EOM with 1HO\ is nonergodic. The origin of the
nonergodicity has been considered as a lack of \textquotedblleft
complexity\textquotedblright, that is, small degrees of freedom of the system
(which is three) and a simple form of the ODE involving only two nonlinear
terms $-\zeta p$\ and $p^{2}$. Its non-ergodicity has been demonstrated
theoretically in case of a sufficiently large $Q$\ by using KAM theory to show
the existence of invariant tori~\cite{Legoll}. In contrast, the NH chain
method, which is an extension of the original NH method via introducing
multidimensional $\zeta\in\mathbb{R}^{m}$, has long been considered that it
gives the ergodicity even in the case of the 1HO. However, recently, Patra and
Bhattacharya indicated in long-time numerical simulations that the NHC\ with
1HO is non ergodic~\cite{Patra}.

However, 1HO is a special model in a viewpoint of the energy density of state
$\Omega(e)$\ for which $\Omega(e)=$constant, which is not increasing with
respect to the energy $e$ of the physical system. In this respect,
$n$-dimensional harmonic oscillator ($n$HO) with $n>1$ is normal in that
$\Omega(e)=ce^{n-1}$ ($c$ is irrelevant to $e$) shows the increasing. Despite
the importance in this respect, NH with $n$HO\ has not been much investigated.
One of a few examples is the study by Nos\'{e}~\cite{Nose1993}, where
recurrence phenomena strongly\ depending on initial conditions were found and
secular periodic modes can be captured by a simple Hamiltonian.

In the current study, we discuss the non ergodicity of $n$HO with $n>1$. We
specifically show in a simple manner that the NH EOM with $n$HO with $n>1$ is
non ergodic if the $n$HO\ is identical. This applies not only to the original
NH system but also for NH type systems, which includes a number of kinds of
EOM such as the NHC equation. The identical condition means that all masses
and the spring constants are identical for all the degrees of freedom. With or
without this condition, symmetry of the system may differ much and the
dynamics of the system can differ quantitatively~\cite{Nose1993}. We also
demonstrate the non ergodicity of NH type with $n$HO with $n>1$\ under a
condition that is slightly extended from the isotropic condition. \ 

The origin of these non ergodicities is underlying symmetries of
multidimensional systems. Here, the original symmetry of the physical system
is the mass spectrum and potential energy function, and they are reflected
into the Newtonian EOM; for example, an interchange $(x_{1},\ldots,x_{i}%
\ldots,x_{j},\ldots,x_{n})\mapsto(x_{1},\ldots,x_{j}\ldots,x_{i},\ldots
,x_{n})$\ becomes one of the symmetries if the system is isotropic. Even if
the EOM\ that is considered to be ergodic for $1$HO, it should be non ergodic
for $n$HO with $n>1$\ (under an isotropic or isotropic-like condition) due to
the existing symmetry, as long as the target EOM has a certain structure
originated from the NH EOM. A number of the thermostat EOM share with this
structure, and the NHC is not the exception, for example. Multidimensional
thermostat EOM is faced with the symmetry and a number of EOM can not be free
from the symmetry, which leads to the nonergodicity in a certain condition.
Hence, the origin of the non ergodicity completely differ between the $1$HO
and $n$HO with $n>1$. These issues will be clarified in the current study. \ 

We further propose a method to recover the ergodicity even for the idealized
model of $n$HO with $n>1$. For this, we first modify the density dynamics (DD)
scheme~\cite{FN2002}, which has been developed to produce an arbitrary
phase-space density. The DD is described by $\left(  x,p,\zeta\right)
$\ where $\zeta$\ is an additional scalar variable as in the NH EOM and plays
a role to control the dynamics to yield the arbitrarily given phase-space
density. We then obtain a new EOM (we call this splitting DD), which utilizes
a vertorized $\zeta$\ and splits the role of the original additional variable
to directly act on each degree of freedom. In general, the splitting DD can be
free from the symmetry discussed above. Next, we apply the BG density, for the
phase-space density, into the splitting DD to produce the BG distribution and
have a new NH form EOM, which we call splitting NH EOM. This new EOM can also
be free from the symmetry through the splitting feature, even if $K(p)$
and$\ U(x)$, the source of the symmetry, are introduced.  

One of the key issues to avoid the problem originated from the symmetry and to
reach the ergodicity in the splitting NH is an extension of the original
Nos\'{e} mass\ parameter $Q$. Nos\'{e} mass is not a scalar but a matrix. This
extension can be viewed as natural in that the density of $\zeta$\ is based on
a quadratic form. Such a matrix approach would be simple, and only a few
discussions have been done~\cite{Samoletov}. The simplicity of this idea might
be the reason that it is has not paid much attention. However, we show that
this idea gives a higher cost performance than supposed to avoid the nonergodicity.

The splitting NH should not be restricted in the applications to the $n$HO
model system. As well as short-rage vibrational interactions, many part of
long range interactions of particles in classical physical system can be
approximated by harmonic interactions around their equilibriums, which can
then be mimicked by $n$HO interactions. More directly, e.g., representations
of the physical system as\ a set of harmonic oscillators in normal mode study
well describe the feature of a biomolecular system. It is discussed that
perturbations by a chemical reaction via enzyme or by a docking of a medicinal
molecule excite normal modes~\cite{Zheng,Dobbins}. These excited normal modes
allow an oscillator approximation such as elastic network model ~\cite{Tirion}%
. Conversely, intramolecular vibrational energy can be transferred from a
given normal mode~\cite{Moritsugu2000}.

Thus, the ability to enhance the phase-space sampling and recover the
ergodicity for the $n$HO with $n>1$\ is not limited to the theoretical
interest but expected to work well in these realistic applications. To solve
the intrinsic problem by breaking the symmetry should truly enhance the
phase-space sampling and reach the equilibrium. We confirmed numerically the
ergodic properties in $n$HO with $n=2,3$ for the splitting NH. \ 

After briefly reviewing thermostat EOM in Section~\ref{Equations of motion},
we demonstrate the nonergodicity for the $n$HO model with $n>1$ in
Section~\ref{Non ergodicity for isotropic oscillator system}. The symmetries
are explicitly discussed in Appendix. In
Section~\ref{Splitting density dynamics}, we provide the splitting DD EOM for
the basis of remedy against the nonergodicity. In
Section~\ref{Splitting Nose-Hoover method}, we apply this to obtain a new NH
form EOM\ to generate the BG distribution and solve the problem. We present
numerical studies in Section~\ref{Numerics}\ to illustrate the nonergodicity
for the conventional EOM and the ergodicity for the new EOM.

\section{Equations of motion\label{Equations of motion}}

Our target dynamical system can be represented as the following ODE:
\begin{equation}
\left.
\begin{array}
[c]{l}%
\dot{x}=\mathbf{M}^{-1}p\in\mathbb{R}^{n},\\
\dot{p}=F(x)+\lambda\left(  \omega\right)  \,p\in\mathbb{R}^{n},\\
\dot{\zeta}=\Lambda(\omega)\in\mathbb{R}^{m},
\end{array}
\right\}  \label{target EOM}%
\end{equation}
where $x\equiv(x_{1},\ldots,x_{n})\in D\subset\mathbb{R}^{n}$ and
$p\equiv(p_{1},\ldots,p_{n})\in\mathbb{R}^{n}$ are atomic coordinates and
momenta of a physical system of $n$\ degrees of freedom; $F$\ represents a
force, which is a $C^{1}$\ function defined on a domain $D;$ $\mathbf{M}%
$\ represents the mass parameters, which is a\ symmetric, positive-definite
square matrix of size $n$ over $\mathbb{R}$.\ Along with these quantities
associated with Newtonian EOM, $\zeta\in\mathbb{R}^{m}$\ is an additional
dynamical variable, relating to a notion of frictional coefficient or
thermostat. Thus the phase space is $\Omega:=D\times\mathbb{R}^{n}%
\times\mathbb{R}^{m}\subset\mathbb{R}^{N}$ with $N\equiv2n+m$, and the
phase-space point is represented as $\omega\equiv\left(  x,p,\zeta\right)
\in\Omega$.\ To the physical system, a $C^{1}$\ function $\lambda$%
\ $:\Omega\rightarrow\mathbb{R}$ provides $-\lambda\left(  \omega\right)  $,
which can be viewed as a dynamical frictional \textquotedblleft
coefficient\textquotedblright\ and essentially depends on the additional
variable $\zeta\in\mathbb{R}^{m}$. The time development of $\zeta$\ is
described by $\Lambda:\Omega\rightarrow\mathbb{R}^{m}$, which is of class
$C^{1}$. The functions $\lambda$\ and $\Lambda$\ may contain potential energy
$U(x)\in\mathbb{R}$, wherein $F=-\nabla U$, and kinetic energy $K(p)\equiv
\frac{1}{2}(p\,|\,\mathbf{M}^{-1}p)\mathbb{=}\frac{1}{2}\sum_{i,j=1}%
^{n}\mathrm{M}_{ij}^{-1}p_{i}p_{j}$ of the physical system .

We give several examples that fall into the form of EOM (\ref{target EOM}).

\begin{example}
\label{NH}$\lambda\left(  \omega\right)  \equiv\zeta/Q\in\mathbb{R}$ and
$\Lambda(\omega)\equiv2K(p)-nk_{\text{B}}T_{\mathrm{ex}}$ with $m\equiv1$,
where $Q>0$\ is\ a parameter (often called as Nos\'{e}'s mass) give the the NH
equation~\cite{N-H eq(1),N-H eq(2)}:
\begin{equation}
\left.
\begin{array}
[c]{l}%
\dot{x}=\mathbf{M}^{-1}p\in\mathbb{R}^{n},\\
\dot{p}=F(x)-(\zeta/Q)\,p\in\mathbb{R}^{n},\\
\dot{\zeta}=2K(p)-nk_{\text{B}}T_{\mathrm{ex}}\in\mathbb{R}^{1}.
\end{array}
\right\}  \label{NH eq}%
\end{equation}
This is introduced to control the physical system temperature
$2K(p)/nk_{\text{B}}$\ ($k_{\text{B}}$ is Boltzmann's constant) into$\ $a
target temperature $T_{\mathrm{ex}}>0$\ and can yield the canonical ensemble.
\end{example}

\begin{example}
\label{NHC} $\lambda\left(  \omega\right)  \equiv\zeta_{1}/Q_{1}\in\mathbb{R}$
and $\Lambda(\omega)\equiv\left(  G_{1}(\omega)-\zeta_{1}\zeta_{2}%
/Q_{2},\ldots,G_{j}(\omega)-\zeta_{j}\zeta_{j+1}/Q_{j+1},\ldots,G_{m}%
(\omega)\right)  $, with
\begin{subequations}
\label{Gj}%
\begin{align*}
G_{1}(\omega)  &  \equiv2K(p)-nk_{\text{B}}T_{\mathrm{ex}}\in\mathbb{R},\\
G_{j}(\omega)  &  \equiv\zeta_{j-1}^{2}/Q_{j-1}-k_{\text{B}}T_{\mathrm{ex}}%
\in\mathbb{R},\ \ j=2,\ldots,m,
\end{align*}
where $Q_{1}\ldots,Q_{m}>0$\ are parameters, give the NH chain (NHC)
equation~\cite{NHC}:
\end{subequations}
\[
\left.
\begin{array}
[c]{l}%
\dot{x}=\mathbf{M}^{-1}p\in\mathbb{R}^{n},\\
\dot{p}=F(x)-(\zeta_{1}/Q_{1})\,p\in\mathbb{R}^{n},\\
\dot{\zeta}_{j}=G_{j}(\omega)-\zeta_{j}\zeta_{j+1}/Q_{j+1}\in\mathbb{R}%
,\ \ j=1,\ldots,m-1,\\
\dot{\zeta}_{m}=G_{m}(\omega)\in\mathbb{R}.
\end{array}
\right\}
\]
This can be viewed as an extended form of the NH EOM.
\end{example}

\begin{example}
The kinetic moments method~\cite{Holian,Harish}, which can be represented by
\[
\left.
\begin{array}
[c]{l}%
\dot{x}=\mathbf{M}^{-1}p\in\mathbb{R}^{n},\\
\dot{p}=F(x)-(\zeta_{1}/Q_{1}+\hat{K}(p)\zeta_{2}/Q_{2})\,p\in\mathbb{R}%
^{n},\\
\dot{\zeta}_{1}=\hat{K}(p)-1\in\mathbb{R},\\
\dot{\zeta}_{2}=\hat{K}(p)(\hat{K}(p)-(n+2)/n)\in\mathbb{R},
\end{array}
\right\}
\]
where $\hat{K}(p)\equiv2K(p)/nk_{\text{B}}T_{\mathrm{ex}}$, becomes an example
of (\ref{target EOM}).
\end{example}

\begin{example}
\label{GGMT}The generalized Gaussian moment thermostatting method~\cite{GGMT}
represented by
\[
\left.
\begin{array}
[c]{l}%
\dot{x}=\mathbf{M}^{-1}p\in\mathbb{R}^{n},\\
\dot{p}=F(x)-\left(  \sum\limits_{j=1}^{m}\sum\limits_{k=1}^{j}a_{k-1}%
(2K(p))^{k-1}(k_{\text{B}}T_{\mathrm{ex}})^{j-k}\zeta_{j}/Q_{j}\right)
\,p\in\mathbb{R}^{n},\\
\dot{\zeta}_{j}=a_{j-1}(2K(p))^{j}-n(k_{\text{B}}T_{\mathrm{ex}})^{j}%
\in\mathbb{R},\ \ j=1,\ldots,m,
\end{array}
\right\}
\]
with $a_{j}\equiv\prod_{k=1}^{j}(n+2k)^{-1}$, becomes also an example.
\end{example}

There can be found more examples in thermostat methods; see e.g.,
Refs.~\cite{Hoover books,Nose
prog,Hunenberger,Jepps,Ezra,Fukuda2016,nonequilibrium work
theorems,Samoletov,Dettmann,Collins,Krajinak} for details on thermostat
methods and their development. Other examples include e.g., the coupled NH
equations of motion, which is introduced to fluctuate the temperature of the
heat bath for the physical system~\cite{FM1,FM2017}.

\section{Non ergodicity for isotropic oscillator system
\label{Non ergodicity for isotropic oscillator system}}

The target ODE~(\ref{target EOM}) can be represented as
\begin{equation}
\dot{\omega}=X(\omega),
\end{equation}
where $X$ becomes a $C^{1}$ vector field defined on a domain $\Omega$\ of
$\mathbb{R}^{N}$. Assuming the completeness of $X$, we let $\{T_{t}%
\}\equiv\{T_{t}:\Omega\rightarrow\Omega$ $|$ $t\in\mathbb{R}\}$ be the flow
generated by the field $X$. Consider the case where we have an invariant
measure $\mu$\ of the flow $\{T_{t}\}$:%
\[
\forall t\in\mathbb{R},\text{ }\forall A\in\mathcal{L}_{N}^{\Omega},\text{
}\mu(T_{t}^{-1}A)=\mu(A),
\]
where $\mathcal{L}_{N}^{\Omega}\equiv\mathcal{L}_{N}\cap\Omega$ with
$\mathcal{L}_{N}$\ being the Lebesgue measurable sets on $\mathbb{R}^{N}$. We
assume that $0<\mu(\Omega)<\infty$ and that
\begin{equation}
\mu\sim l_{N},\label{equiv of measures}%
\end{equation}
i.e., $\mu$\ and the Lebesgue measure $l_{N}$\ of $\mathbb{R}^{N}$ are
absolutely continuous each other. For example, on Examples~\ref{NH}
and~\ref{NHC}, we have an invariant measure defined by%
\begin{equation}
\mathcal{L}_{N}^{\Omega}\rightarrow\lbrack0,\infty),\text{ }A\mapsto
\mu(A):=\int_{A}\rho dl_{N}\label{myu}%
\end{equation}
with a (strictly positive and measurable) density $\rho:\Omega\rightarrow
\mathbb{R}$, satisfying the Liouville equation~\cite{comment,FM2}%
\begin{equation}
\operatorname{div}\rho X=0.\label{Liouville equation}%
\end{equation}

A subset $A\in\mathcal{L}_{N}^{\Omega}$ is said to be an invariant set if
$T_{t}^{-1}(A)=A$ for all $t\in\mathbb{R}$. The ergodicity for the measure
space $(\Omega,\mathcal{L}_{N}^{\Omega},\mu)$\ with the flow $\{T_{t}\}$ holds
if any invariant set $A$\ is trivial, i.e., $\mu(A)=0$ or $\mu(\Omega
\backslash A)=0$. In other word, if we have an invariant set $A$\ such that
\begin{equation}
\mu(A)>0\text{ and }\mu(\Omega\backslash A)>0,\label{non ergode condi}%
\end{equation}
then the ergodicity does not hold.

We show $\{T_{t}\}$ is not ergodic for a harmonic oscillator system with
$n>1$. The condition for $n>1$\ is essential to the current discussion. The
non ergodicity for a harmonic oscillator system with $n=1$ is demonstrated in
Ref.~\cite{Legoll}. Here, an isotropic condition in a harmonic oscillator
system in Eq.~(\ref{target EOM}) is described by
\begin{equation}
\mathbf{M=}m\mathbf{1}_{n}\text{ and }F(x)=-kx\in\mathbb{R}^{n},
\label{isotropic condi}%
\end{equation}
where $\mathbf{1}_{n}$\ is the unit matrix of size $n$, and $m$\ and $k$ are
strictly positive parameters (representing a mass and spring constant,
respectively); i.e., we have
\begin{equation}
\left.
\begin{array}
[c]{l}%
\dot{x}=m^{-1}p\in\mathbb{R}^{n},\\
\dot{p}=-kx+\lambda\left(  \omega\right)  \,p\in\mathbb{R}^{n},\\
\dot{\zeta}=\Lambda(\omega)\in\mathbb{R}^{m}.
\end{array}
\right\}  \label{identical nHO EOM}%
\end{equation}
Defining a map $\gamma:\Omega\rightarrow T^{2}(\mathbb{R}^{n})\cong%
\mathbb{R}^{n^{2}}$ by
\begin{align}
\gamma(\omega)  &  :=x\wedge p\nonumber\\
&  =\frac{1}{2}(x\otimes p-p\otimes x), \label{def of gamma}%
\end{align}
we show

\begin{lemma}
\label{L1}$\Upsilon_{ij}^{\pm}\equiv\{\omega\in\Omega$ $|$ $\gamma_{ij}%
(\omega)\gtrless0\}$ is an invariant space for the flow of
Eq.~(\ref{identical nHO EOM}) for $i,j=1,\ldots,n$.
\end{lemma}

\begin{proof}
For any solution of ODE~(\ref{identical nHO EOM}), $\varphi:\mathbb{R\supset
}J$ $\rightarrow\Omega,t\mapsto\varphi(t)\equiv\left(  x(t),p(t),\zeta
(t)\right)  ,$ we have
\begin{align*}
D(\gamma\circ\varphi)(t)  &  =Dx(t)\wedge p(t)+x(t)\wedge Dp(t)\\
&  =m^{-1}p(t)\wedge p(t)+x(t)\wedge(-kx(t)+\lambda\left(  \varphi(t)\right)
p(t))\\
&  =x(t)\wedge\lambda(\varphi(t))p(t)\\
&  =\lambda(\varphi(t))(\gamma\circ\varphi)(t)
\end{align*}
for all $t\in J$, which is an open interval (that may be $\mathbb{R}$)
involving $0$. Thus
\[
(\gamma\circ\varphi)(t)=\exp\left(  \int_{0}^{t}\lambda(\varphi(s))ds\right)
(\gamma\circ\varphi)(0)\in\mathbb{R}^{n^{2}}%
\]
for any $t\in J$, implying that $(\gamma\circ\varphi)(0)=0\ $reads as
$(\gamma\circ\varphi)(t)=0\ $for all $t$.$\ $Hence%
\[
\Upsilon_{ij}^{0}\equiv\{\omega\in\Omega|\gamma_{ij}(\omega)=0\}
\]
is an invariant space for $i,j=1,\ldots,n$. For any $i$ and $j$, the
continuity of $\gamma_{ij}$\ indicates that $\Upsilon_{ij}^{+}=\{\gamma
_{ij}>0\}$ and $\Upsilon_{ij}^{-}=\{\gamma_{ij}<0\}$ are also invariant. \ 
\end{proof}

Thus, we get

\begin{proposition}
\label{nonergodicity for nHO}The flow $\{T_{t}\}$\ of
ODE~(\ref{identical nHO EOM}) is not ergodic with respect to $(\Omega
,\mathcal{L}_{N}^{\Omega},\mu)$.
\end{proposition}

\begin{proof}
Choose any $i$,$j\in\{1,\ldots,n\}$ such that $i\neq j$ (recall $n>1$). We
have $\Upsilon_{ij}^{\pm}\neq\emptyset$. From
assumption~(\ref{equiv of measures}) and the fact that $\Upsilon_{ij}^{\pm}%
$\ becomes a nonempty open set of $\mathbb{R}^{N}$, we see $\mu(\Upsilon
_{ij}^{+})>0$ and $\mu(\Omega\backslash \Upsilon_{ij}^{+})\geq\mu(\Upsilon
_{ij}^{-})>0$.
\end{proof}

The above discussion shows the non ergodicity for the harmonic oscillator
system with $n>1$, by showing the existence of an invariant set $\Upsilon
_{ij}^{+}$ that has a desired property (for this purpose, finding one such a
set is relevant but the existence of many subsets such as $\cup_{i,j}%
\Upsilon_{ij}^{\pm}$\ and $\cap_{i,j}\Upsilon_{ij}^{\pm}$\ is less important).
However, this discussion lacks explanations why and how the invariant set
$\Upsilon_{ij}^{\pm}$\ arises. In Appendix, we show that symmetry group
$O(N)$\ acts on ODE~(\ref{identical nHO EOM}) and its invariant set splits the
total phase space $\Omega$ producing invariant sets that correspond to
$\Upsilon_{ij}^{\pm}$. It should be noted that finding just one solution (or
countably many solutions) confined in a certain subset $B$\ does not
necessarily indicate the non ergodicity, since $B$\ may be a null set.

The above discussion on the isotropic harmonic oscillator\ case can be
generalized to a harmonic oscillator\ with $F(x)=-\mathbf{K}x\in\mathbb{R}%
^{n}$ with $\mathbf{K}\in\mathrm{End}\mathbb{R}^{n}$\ under the condition that
$\mathbf{M}^{-1}\mathbf{K}$\ is symmetric, positive definite, and\ degenerate,
i.e., there exist eigenvalues such that $\lambda_{i}=\lambda_{j}$ for $i\neq
j$. Proposition~\ref{nonergodicity for nHO} in this case is proven by
observing that $\tilde{\Upsilon}_{ij}^{0}\equiv\{\omega\in\Omega|\tilde
{\gamma}_{ij}(\omega)=0\}$ becomes an invariant space, where $\tilde{\gamma
}(\omega)\equiv VGx\wedge G\mathbf{M}^{-1}p$\ with $V\equiv G\mathbf{M}%
^{-1}\mathbf{K}G^{-1}$\ being the diagonal matrix for a certain $G\in$GL$(n)$,
and by observing that $\tilde{\Upsilon}_{ij}^{\pm}\equiv\{\omega\in
\Omega|\tilde{\gamma}_{ij}(\omega)\gtrless0\}$ becomes a nonempty open
invariant set.

\section{Splitting density dynamics\label{Splitting density dynamics}}

To attain the ergodicity for symmetric systems such as the $n$HO, we propose a
generalized version of the density dynamics (DD)~\cite{FN2002}. The original
version of the DD is defined by
\begin{equation}
\left.
\begin{array}
[c]{l}%
\dot{x}_{i}=D_{p_{i}}\Theta(\omega)\in\mathbb{R},\ \ \ i=1,\ldots,n,\\
\dot{p}_{i}=-D_{x_{i}}\Theta(\omega)-D_{\zeta}\Theta(\omega)\ p_{i}%
\in\mathbb{R},\ \ \ i=1,\ldots,n,\\
\dot{\zeta}=\sum\limits_{i=1}^{n}D_{p_{i}}\Theta(\omega)\ p_{i}-n\beta^{-1}%
\in\mathbb{R},
\end{array}
\right\}  \label{DD eq}%
\end{equation}
with $\Theta=-\beta^{-1}\ln\rho$, where $\rho:\Omega\rightarrow\mathbb{R}$ is
an arbitrarily given density function, i.e., $\rho$ is a function that is of
class $C^{2}$, strictly positive, and integrable. $\zeta\in\mathbb{R}$ is a
dynamical variable and $\beta>0$\ is an arbitrary parameter.\ ODE~(\ref{DD eq}%
) is designed so as to satisfy the Liouville
equation~(\ref{Liouville equation}) for the density $\rho$, and it involves
the NH EOM (\ref{NH eq}), viz., the NH\ is recovered if we set $\rho
=\rho_{\text{NH}}$, where
\begin{equation}
\rho_{\text{NH}}\left(  \omega\right)  \equiv\exp\left[  -\beta\left(
U(x)+K(p)+\frac{1}{2Q}\zeta^{2}\right)  \right]  \label{NH ro}%
\end{equation}
with $\beta=1/k_{\text{B}}T_{\mathrm{ex}}$.

Our generalization for Eq.~(\ref{DD eq}) is based on (i) an extension of the
additional scalar variable $\zeta\in\mathbb{R}$\ to a vector variable
$\zeta\equiv(\zeta_{1},\ldots,\zeta_{n})\in\mathbb{R}^{n}$, and (ii)
$\zeta_{i}$\ 's EOM that is a natural decomposition of the third equation of
~(\ref{DD eq}). Namely, a generalized DD, which we call splitting density
dynamics, is
\begin{equation}
\left.
\begin{array}
[c]{l}%
\dot{x}_{i}=D_{p_{i}}\Theta(\omega)\in\mathbb{R},\ \ \ i=1,\ldots,n,\\
\dot{p}_{i}=-D_{x_{i}}\Theta(\omega)-D_{\zeta_{i}}\Theta(\omega)\ p_{i}%
\in\mathbb{R},\ \ \ i=1,\ldots,n,\\
\dot{\zeta}_{i}=D_{p_{i}}\Theta(\omega)\ p_{i}-\beta^{-1}\in\mathbb{R}%
,\ \ \ i=1,\ldots,n.
\end{array}
\right\}  \label{massive DD EOM}%
\end{equation}
As is easily confirmed that the Liouville equation~(\ref{Liouville equation})
holds for any density function $\rho$, this new EOM can be replaced with
Eq.~(\ref{DD eq}). That is, for any $P$-integrable function $g$ on phase space
$\Omega$,%
\begin{align}
\bar{g} &  :=\exists\lim\limits_{\tau\rightarrow\infty}\dfrac{1}{\tau}%
{\displaystyle\int_{0}^{\tau}}
g(T_{t}(\omega))dt\nonumber\\
&  =\int_{\Omega}g\rho\text{ }dl_{N}\left/  \int_{\Omega}\rho dl_{N}\right.
=:\left\langle g\right\rangle \in\mathbb{R}\label{ergodic prop}%
\end{align}
holds with respect to a $P$-almost every initial point $\omega$, if the flow
$\{T_{t}\}$ is ergodic with respect to the measure $P\equiv\rho dl_{N}$.

Some remarks are made~\cite{FN2002}. First, fixed points for $X,$ which can be
obstructions to the ergodicity, do not exist, as long as $\beta>0$. Second,
$\operatorname*{div}X\neq0$ holds (otherwise $\rho$ becomes an invariant
function and should not be almost everywhere constant, so the system does not
become ergodic), as long as
\begin{equation}
\sum_{i=1}^{n}D_{\zeta_{i}}\rho\neq0 \label{condi of ro  for div}%
\end{equation}
(not identically zero). This is because $\operatorname*{div}X=-\sum_{i=1}%
^{n}D_{\zeta_{i}}\Theta$. Condition~(\ref{condi of ro for div}) is valid in
many cases (see the case later).

\section{Splitting Nos\'{e}-Hoover method utilizing Nos\'{e}-mass matrix
\label{Splitting Nose-Hoover method}}

The meaning of the generalization of (i) and (ii) in
Section~\ref{Splitting density dynamics} will be clearer when we consider the
NH limit. Here, the NH limit is obtained if we set $\rho=\tilde{\rho
}_{\text{BG}}$, where
\begin{equation}
\tilde{\rho}_{\text{BG}}\left(  \omega\right)  \equiv\exp\left[  -\beta\left(
U(x)+K(p)+K_{z}(\zeta)\right)  \right]  \label{new NH ro}%
\end{equation}
with
\begin{equation}
K_{z}(\zeta)\equiv\frac{1}{2}(\zeta\,|\,\mathbf{Q}^{-1}\zeta)=\frac{1}{2}%
\sum_{i,j=1}^{n}\mathrm{Q}_{ij}^{-1}\zeta_{i}\zeta_{j}%
\end{equation}
being a quantity corresponding to the kinetic energy for $\zeta$, and
$\beta=1/k_{\text{B}}T_{\mathrm{ex}}$. The difference from Eq.~(\ref{NH ro})
is to utilize, as well as the vectorized $\zeta\in\mathbb{R}^{n}$, a matrix
form of $\mathbf{Q}$, which is a natural extension of the original scalar
Nos\'{e}'s mass\ $Q$ (recovered when $n=1$, of course). Specifically,
$\mathbf{Q\equiv(}\mathrm{Q}_{ij}\mathbf{)}\in\mathrm{End}\mathbb{R}^{n}%
$\ should be symmetric and positive definite: we should set it symmetric
without loss of generality, considering that the kinetic energy is a quadratic
form; and the positive-definite condition is a natural extension of $Q>0$\ and
is actually required for ensuring the integrability condition of $\rho$. Now,
applying Eq.~(\ref{new NH ro}), the splitting DD~(\ref{massive DD EOM}) turns
out to be
\begin{subequations}
\label{massive NH EOM}%
\begin{align}
\dot{x}_{i}  &  =(\mathbf{M}^{-1}p)_{i}\in\mathbb{R},\ \ \ i=1,\ldots
,n,\label{massive NH EOM 1}\\
\dot{p}_{i}  &  =F_{i}(x)-\tau_{i}\left(  \zeta\right)  \ p_{i}\in
\mathbb{R},\ \ \ i=1,\ldots,n,\label{massive NH EOM 2}\\
\dot{\zeta}_{i}  &  =2K_{i}(p)-k_{\text{B}}T_{\mathrm{ex}}\in\mathbb{R}%
,\ \ \ i=1,\ldots,n, \label{massive NH EOM 3}%
\end{align}
where
\end{subequations}
\begin{align}
\tau_{i}\left(  \zeta\right)   &  \equiv-k_{\text{B}}T_{\mathrm{ex}}%
D_{\zeta_{i}}\ln\tilde{\rho}_{\text{BG}}\left(  \omega\right) \nonumber\\
&  =D_{i}K_{z}(\zeta)\nonumber\\
&  =(\mathbf{Q}^{-1}\zeta)_{i}=\sum_{j=1}^{n}\mathrm{Q}_{ij}^{-1}\zeta_{j}
\label{frictional term of the current method}%
\end{align}
and
\begin{equation}
2K_{i}(p)\equiv(\mathbf{M}^{-1}p)_{i}\ p_{i}.
\end{equation}
Thus, the dynamical frictional \textquotedblleft coefficient\textquotedblright%
\ $\tau_{i}\left(  \zeta\right)  $\ depends on not only one component for
$\zeta$\ (as in the NH and NHC) but also all components $\zeta_{1}%
,\ldots,\zeta_{n}$\ or at least two components $\zeta_{k},\zeta_{l}$\ when we
choose $\mathbf{Q}$ as a non-diagonal matrix.\ This is a motivation of above
(i), and this \textquotedblleft mixing\textquotedblright\ of $\zeta
$\ components will play a part for avoiding the nonergodicity, as detailed
below. Briefly speaking, the fact that the contribution of $\tau_{i}\left(
\zeta\right)  $\ to $\dot{p}_{i}$\ can be different for each $i$ is effective
to break down the isotropic symmetry if it exists in the system.
Equation~(\ref{massive NH EOM 3})\ intends the law of equipartition, i.e., the
expected equilibrium condition, $\dot{\zeta}_{i}\sim0$, should contribute to
the exact relationship of the law $\overline{2K_{i}}=k_{\text{B}%
}T_{\mathrm{ex}}$ for every degree of freedom $i$, which is validated by%
\begin{align}
\overline{2K_{i}}  &  =\left\langle 2K_{i}\right\rangle \nonumber\\
&  \equiv\int_{\Omega}(\mathbf{M}^{-1}p)_{i}\ p_{i}\tilde{\rho}_{\text{BG}%
}(\omega)dl_{N}(\omega)\left/  \int_{\Omega}\tilde{\rho}_{\text{BG}}%
dl_{N}\right. \nonumber\\
&  =k_{\text{B}}T_{\mathrm{ex}} \label{equipartition}%
\end{align}
where the first equation is owing to ~(\ref{ergodic prop}) under the ergodic condition.

Some remarks are made. First note that the EOM ~(\ref{massive NH EOM}) is
still physically natural in the sense that the first equation is exactly same
as that in the Newtonian EOM, and the $i$th component of the frictional force
in Eq.~(\ref{massive NH EOM 2}) is proportional to $p_{i}$ and take a form
$-\tau_{i}\left(  \zeta\right)  \ p_{i}$\ using a scalar quantity $\tau
_{i}\left(  \zeta\right)  \in\mathbb{R}$, which conforms to the conventional
form for the classical dynamics treatment. Second, EOM ~(\ref{massive NH EOM})
can be viewed\ as a generalization of the original NH. This is because, by
setting
\begin{equation}
\mathbf{Q}^{-1}=Q^{-1}\boldsymbol{\text{\b{1}}}\equiv Q^{-1}\left[
\begin{array}
[c]{ccc}%
1 & \cdots & 1\\
\vdots & \ddots & \vdots\\
1 & \cdots & 1
\end{array}
\right]  \in\mathrm{End}\mathbb{R}^{n}\label{uniform matrix}%
\end{equation}
with a scalar input $Q^{-1}>0$\ and a matrix $\mathbf{\text{\b{1}}}$\ whose
every element is unity (not the unit matrix), we recover the original NH
EOM\ via defining a redesigned scalar variable $\zeta\equiv\sum_{i=1}^{n}%
\zeta_{i}\in\mathbb{R}$, which plays the same role in the original NH
variable. In this sense, we will call Eq.~(\ref{massive NH EOM}) splitting
Nos\'{e}-Hoover EOM. Third, as can also be seen from Eq.~(\ref{equipartition}%
), ODE~(\ref{massive NH EOM})\ generates the BG distribution at temperature
$T_{\mathrm{ex}}$\ for physical quantities under the ergodic condition. This
is clearly seen by separately rewriting $\tilde{\rho}_{\text{BG}}$ such that
\begin{align*}
\tilde{\rho}_{\text{BG}}\left(  \omega\right)   &  =\rho_{\text{BG}}\left(
x,p\right)  \rho_{\text{Z}}\left(  \zeta\right)  ,\\
\rho_{\text{BG}}\left(  x,p\right)   &  \equiv\exp\left[  -\frac
{1}{k_{\text{B}}T_{\mathrm{ex}}}\left(  U(x)+K(p)\right)  \right]  ,\\
\rho_{\text{Z}}\left(  \zeta\right)   &  \equiv\exp\left[  -\frac
{1}{k_{\text{B}}T_{\mathrm{ex}}}K_{z}(\zeta)\right]  ,
\end{align*}
Eq.~(\ref{ergodic prop}) implies the relation $\bar{f}=\left\langle
f\right\rangle _{\text{BG}}$ with respect to a physical quantity
$f:D\times\mathbb{R}^{n}\rightarrow\mathbb{R}$\ such that
\begin{align}
\bar{f} &  \equiv\lim\limits_{\tau\rightarrow\infty}\dfrac{1}{\tau}%
{\displaystyle\int_{0}^{\tau}}
f(x(t),p(t))dt\nonumber\\
&  =\int_{\Omega}f(x,p)\tilde{\rho}_{\text{BG}}(\omega)\text{ }dl_{N}\left/
\int_{\Omega}\tilde{\rho}_{\text{BG}}dl_{N}\right.  \in\mathbb{R}\nonumber\\
&  =\frac{\int_{D\times\mathbb{R}^{n}}f(x,p)\rho_{\text{BG}}\left(
x,p\right)  dxdp}{\int_{D\times\mathbb{R}^{n}}\rho_{\text{BG}}\left(
x,p\right)  dxdp}=:\left\langle f\right\rangle _{\text{BG}}%
.\label{PS variable average}%
\end{align}
Finally, note that Condition~(\ref{condi of ro for div})\ is valid for
$\rho=\tilde{\rho}_{\text{BG}}$ since $\sum_{i=1}^{n}D_{\zeta_{i}}%
\Theta\left(  \zeta\right)  =\sum_{i,j=1}^{n}\mathrm{Q}_{ij}^{-1}\zeta_{j}$.

Note that matrix form of $\mathbf{Q}$\ has been utilized before in literature.
Samoletov et al.~\cite{Samoletov} have used it in their development of
configurational thermostats, which are thermostat equations in configuration
space. Their matrix $\mathbf{Q}$\ is of size of $3$, which corresponds to the
additional $3$-vector introduced for their purpose in order to control
$x$\ instead of $p$. They also utilized a diagonal form (uncoupled case) for
$\mathbf{Q}$ in considering physical necessities,\ though admitted the
possibility of coupled case. In these respects, their approach and ours are
different. Leimkuhler et al.~\cite{Leimkuhler}\ has considered a NH type EOM
with introducing random noise to improve ergodicity. The frictional term they
treated is of form of $\lambda\left(  \omega\right)  =\zeta\mathbf{1}%
_{n}+\mathbf{M\,S}(t,\zeta)\in\mathrm{End}\mathbb{R}^{n}$, where $\zeta$\ is a
scalar as in the original NH, and $\mathbf{S}(t,\zeta)$\ is an anti-symmetric
matrix depending on $\zeta$.\ Thus, it is different from our
term~(\ref{frictional term of the current method}). They demonstrated the
ergodicity for their stochastic dynamics with harmonic oscillators, relating
as a counterpart to our statement of the nonergodicity for ODE.

\subsection{On the choice of matrix $\mathbf{Q}$}

We describe how the matrix $\mathbf{Q}$ defines the the distribution of
$\zeta$ and how we should set $\mathbf{Q}$\ for effectively realize the ergodicity.

\subsubsection{$\mathbf{Q}$ determines $\zeta$'s distribution}

In contrast to the (marginal) distribution of $\left(  x,p\right)
$\ described by the RHS of Eq.~(\ref{PS variable average}), which is the BG
distribution, the distribution of $\zeta$ is characterized by the matrix
$\mathbf{Q}$\ and is described by $P_{\zeta}\equiv P\pi_{\zeta}^{-1}%
:\mathcal{B}^{n}\rightarrow\mathbb{R}$ such that
\begin{align*}
B &  \mapsto P(\pi_{\zeta}^{-1}(B))\\
&  =\int_{D\times\mathbb{R}^{n}\times B}\tilde{\rho}_{\text{BG}}dl_{N}\left/
\int_{\Omega}\tilde{\rho}_{\text{BG}}dl_{N}\right.  \\
&  =\int_{B}\rho_{\text{Z}}\left(  \zeta\right)  dl_{n}\left(  \zeta\right)
\left/  \int_{\mathbb{R}^{n}}\rho_{\text{Z}}dl_{n}\right.  \\
&  =N_{z}\int_{B}\exp\left[  -\frac{1}{2k_{\text{B}}T_{\mathrm{ex}}}%
(\zeta\,|\,\mathbf{Q}^{-1}\zeta)\right]  d\zeta,
\end{align*}
where $N_{z}\equiv\lbrack(2\pi k_{\text{B}}T_{\mathrm{ex}})^{n}\det
\mathbf{Q}]^{-1/2}$. Namely, $\zeta\in\mathbb{R}^{n}$\ is distributed
ellipsoidally\ around the origin. Note that, instead of directly using
$P_{\zeta}$, it is often connivent to use the distribution of principal
component $y\equiv O^{-1}\zeta\in\mathbb{R}^{n}$ for which $\mathbf{Q}^{-1}$
is diagonalized as $O^{-1}\mathbf{Q}^{-1}O=\mathrm{diag}(\lambda_{1}%
,\ldots,\lambda_{n})$\ with $\lambda_{i}$\ being a strictly positive
eigenvalue of $\mathbf{Q}^{-1}$. The distribution of $y$ is given by
$P_{Y}\equiv P(G^{-1}\circ\pi_{\zeta})^{-1}:\mathcal{B}^{n}\rightarrow
\mathbb{R}$ where
\begin{align}
B &  \mapsto P(\pi_{\zeta}^{-1}(G(B)))\nonumber\\
&  =N_{z}\int_{B}\exp\left[  -\frac{1}{2k_{\text{B}}T_{\mathrm{ex}}}\sum
_{i=1}^{n}\lambda_{i}y_{i}^{2}\right]  dy,\label{distribution of y}%
\end{align}
which is the joint distribution of 1-dimensional Gaussian distributions
$\exp\left[  -\frac{1}{2k_{\text{B}}T_{\mathrm{ex}}}\lambda_{i}y_{i}%
^{2}\right]  dy_{i}$, $i=1,\ldots,n$. Note that the distribution of $\zeta$,
or $y=O^{-1}\zeta$, is not used in obtaining physical information, such as the
long-time average of physical variable in Eq.~(\ref{PS variable average}), but
the explicit form of $P_{Y}$\ can be utilized to monitor the convergence of
the distribution generated by the flow and numerically judge the ergodicity.

\subsubsection{We determine $\mathbf{Q}$\textbf{\label{We determine Q}}}

Here we discuss how we determine $\mathbf{Q}$\ or $\mathbf{Q}^{-1}$. Its
overall amplitude can be set by a scale factor\ as in the case of the scalar
$Q$, as in\ the original NH~\cite{Nose prog}. Thus we should determine the
difference between the matrix elements it in a finer manner. Our criteria for
setting the matrix $\mathbf{Q}^{-1}$\ are as follows:\ 

\begin{itemize}
\item[(i)] it is\ symmetric and positive definite;

\item[(ii)] its eigenvalues and eigen vectors are explicitly obtained;

\item[(iii)] it should not be diagonal;

\item[(iv)] its diagonal components are nevertheless sufficiently larger than
off-diagonal components;

\item[(v)] randomness can be easily introduced in the elements.\ 
\end{itemize}

The reason of these requirements is as follows: (i) has been already assumed
and the necessity has also been discussed. (ii) is required to explicitly
obtain the distribution of $\zeta$\ or $y\equiv G^{-1}\zeta$. (iii) is needed
to enhance the mixing of different components $\zeta_{1},\ldots,$ and
$\zeta_{n}$\ through the friction term%
\begin{equation}
-\tau_{i}\left(  \zeta\right)  p_{i}=-\left(  \mathrm{Q}_{i1}^{-1}\zeta
_{1}+\cdots+\mathrm{Q}_{ii}^{-1}\zeta_{i}+\cdots\mathrm{Q}_{in}^{-1}\zeta
_{n}\right)  p_{i}%
\end{equation}
in Eq.~(\ref{massive NH EOM 2}). Otherwise, Eq.~(\ref{massive NH EOM 2}) turns
out to be the same form as that of the original NH, leading to the
nonergodicity in the case of the $n$HO as discussed in
Section~\ref{Non ergodicity for isotropic oscillator system}. (iv) if
$\mathrm{Q}_{ii}^{-1}$\ is small, then the contribution of $\zeta_{i}$ derived
by Eq.~(\ref{massive NH EOM 3}) will not be much assessed, so that the
equipartition~(\ref{equipartition})\ will not be enhanced at least in a
relatively short time scale. (v) is needed to break isotropy or symmetries in
the target physical system. It is also useful to emphasize the difference
between the splitting NH\ and the original NH, where the latter can be
characterized as a uniform matrix $\mathbf{Q}^{-1}$ seen
in~(\ref{uniform matrix}).

Using the fact that a symmetric matrix $\mathbf{W}\in\mathrm{End}%
\mathbb{R}^{n}$\ is positive definite\ if and only if $\exists O\in O(n)$,
$\exists d_{1},\ldots,d_{n}>0$, $\mathbf{W}=O\,\mathrm{diag}(d_{1}%
,\ldots,d_{n})\mathbf{\,}^{\text{T}}O$, and using a representation of the
group $O(n)$,\ we propose the following procedures (1)--(4) for setting
$\mathbf{Q}^{-1}$:

(1) Choose values randomly for $\theta_{k,j}\in$ $]0,\epsilon\lbrack$\ with
$0<\epsilon\ll\pi$ for $1\leq j<k\leq n$,\ 

(2) define $O:=h_{n}h_{n-1}\cdots h_{2}$ for which $h_{k}:=r_{1}(\theta
_{k,1})r_{2}(\theta_{k,2})\cdots r_{k-1}(\theta_{k,k-1})$, where\
\[
r_{i}(\theta)\equiv\left[
\begin{array}
[c]{ccc}%
\mathbf{1}_{i-1} & \mathbf{0} & \mathbf{0}\\
\mathbf{0} & u_{2}(\theta) & \mathbf{0}\\
\mathbf{0} & \mathbf{0} & \mathbf{1}_{n-i-1}%
\end{array}
\right]  \in\mathrm{End}\mathbb{R}^{n}%
\]
with $\mathbf{1}_{i}$\ being the unit matrix of size $i$ and $u_{2}%
(\theta)\equiv\left[
\begin{array}
[c]{cc}%
\cos{\small \theta} & \sin{\small \theta}\\
{\small -}\sin{\small \theta} & \cos{\small \theta}%
\end{array}
\right]  $; for convenience, any $q\in O(n)$, such as $q\equiv\mathrm{diag}%
(1,\ldots,1,-1)$\ or interchange matrices $q\equiv J_{ij}$ can be inserted
among the products of $h_{k}$\ in defining $O$,

(3) set $d_{i}=1+\delta_{i}$ with $-\delta<\delta_{i}<\delta<1$\ for
$i=1,\ldots,n$, where $\delta_{i}\neq\delta_{j}$\ for $i\neq j$, and then,
using a scale factor $\lambda$, put $\mathbf{D}\equiv\lambda\,\mathrm{diag}%
(d_{1},\ldots,d_{n})\equiv:\mathrm{diag}(\lambda_{1},\ldots,\lambda_{n})$, and finally,

(4) define
\begin{equation}
\mathbf{Q}^{-1}:=O\mathbf{\,D\,}^{\text{T}}O.
\end{equation}

Then, condition (i) holds, and (ii) is clear since the eigen values\ of
$\mathbf{Q}^{-1}$ are $\lambda_{1},\ldots,\lambda_{n}$\ and $O_{i}%
\equiv\mathbf{\,}^{\text{T}}(O_{1i},\ldots,O_{ni})\in\mathbb{R}^{n}$\ is
obtained to be the eigen vector\ corresponding to $\lambda_{i}$ for
$i=1,\ldots,n$. Condition (iii) will hold, or reset some values of
$\theta_{k,j}$ if needed. A small $\epsilon$ that becomes the threshold of
$\theta_{k,j}$ is useful to contribute to the purpose (iv) in that each
$r_{i}(\theta)$\ is near the identity matrix. Randomness can be introduced
through the $n(n-1)/2$\ manifold parameters $\theta_{k,j}$ for the sake of (v).

\begin{example}
For $n=2$, we have $O=h_{2}=r_{1}(\theta_{2,1})=u_{2}(\theta_{2,1})$. For
$n=3$, we have \
\begin{align*}
O  &  =h_{3}h_{2}=r_{1}(\theta_{3,1})r_{2}(\theta_{3,2})r_{1}(\theta_{2,1})\\
&  =\left[
\begin{array}
[c]{ccc}%
\cos\theta_{3,1} & \sin\theta_{3,1} & 0\\
{\small -}\sin\theta_{3,1} & \cos\theta_{3,1} & 0\\
{\normalsize 0} & {\normalsize 0} & {\normalsize 1}%
\end{array}
\right]  \left[
\begin{array}
[c]{ccc}%
1 & 0 & 0\\
0 & \cos\theta_{3,2} & \sin\theta_{3,2}\\
{\normalsize 0} & {\small -}\sin\theta_{3,2} & \cos\theta_{3,2}%
\end{array}
\right]  \left[
\begin{array}
[c]{ccc}%
\cos\theta_{2,1} & \sin\theta_{2,1} & 0\\
{\small -}\sin\theta_{2,1} & \cos\theta_{2,1} & 0\\
{\normalsize 0} & {\normalsize 0} & {\normalsize 1}%
\end{array}
\right]  .
\end{align*}
Instead, we can use e.g., $O=J_{12}r_{1}(\theta_{3,1})J_{12}J_{12}r_{2}%
(\theta_{3,2})J_{12}J_{13}r_{1}(\theta_{2,1})J_{13}$, viz.,%
\begin{equation}
O=\left[
\begin{array}
[c]{ccc}%
\cos\theta_{3,1} & -\sin\theta_{3,1} & 0\\
\sin\theta_{3,1} & \cos\theta_{3,1} & 0\\
{\normalsize 0} & {\normalsize 0} & {\normalsize 1}%
\end{array}
\right]  \left[
\begin{array}
[c]{ccc}%
\cos\theta_{3,2} & 0 & \sin\theta_{3,2}\\
0 & 1 & 0\\
{\small -}\sin\theta_{3,2} & 0 & \cos\theta_{3,2}%
\end{array}
\right]  \left[
\begin{array}
[c]{ccc}%
1 & 0 & 0\\
0 & \cos\theta_{2,1} & -\sin\theta_{2,1}\\
{\normalsize 0} & \sin\theta_{2,1} & \cos\theta_{2,1}%
\end{array}
\right]  , \label{alt scheme for def of O}%
\end{equation}
which indicates the composition of rotations in $\mathbb{R}^{3}$\ around the
x,y,z-axis with angles $\theta_{2,1}$, $\theta_{3,2}$, and $\theta_{3,1}$, respectively.\ 
\end{example}

\section{Numerics\label{Numerics}}

We numerically tested our considerations, nonergodic property for the
conventional schemes and the ergodic property for the current scheme, using
the isotropic harmonic oscillator system defined by (\ref{isotropic condi})
with $n>1$. We set both the mass $m$\ and spring constant $k$\ to be $1$, and
put $k_{\text{B}}T_{\mathrm{ex}}=1$ (all quantities were treated as
dimensionless). Numerical integrations of ODEs\ were done by the explicit
second order scheme described in Ref.~\cite{Fukuda2019} for $10^{8}$\ time
steps with a unit time of $h=10^{-3}$, and the numerical errors were checked
to be within a tolerance within the scheme of the extended
system~\cite{PRE2006}.

We show that conventional method employing an EOM\ of the form of
Eq.~(\ref{identical nHO EOM}) fail in the ergodic sampling, as stated in
Proposition~\ref{nonergodicity for nHO}. As a conventional method, we have
used the NHC method (Example~\ref{NHC}) with the chain length $m=2$ and masses
$Q_{1}=Q_{2}=1$. A first case we show is that with $n=2$, where we used the
initial value $x(0)=(0,0),$ $p(0)=(1,1),$ and $\zeta(0)=(0,0)$. Figure 1 shows
that the trajectory of $(x_{1},p_{1})$\ and their marginal distributions. The
trajectory shows a hall in a vicinity of the origin, and this clearly affects
the distribution of $x_{1}$.\ The distribution of $p_{1}$\ is also weird. Due
to the a special setting of the initial condition, $x_{1}(t)=x_{2}(t)$\ and
$p_{1}(t)=p_{2}(t)$ for all time $t$, so that the trajectory of $(x_{2}%
,p_{2})$\ and the distributions\ are totally the same as that for
$(x_{1},p_{1}),$\ respectively. In terms of the symmetry, this special initial
condition obeys a symmetry of the interchange, $S=\left[
\begin{array}
[c]{cc}%
0 & 1\\
1 & 0
\end{array}
\right]  \in O(2)$, so that the initial value\ $\omega_{0}\equiv\left(
x(0),p(0),\zeta(0)\right)  $ falls in an invariant set $A=A_{\mathcal{R}%
}=A_{O(2)_{\times}}$, defined in Eq.~(\ref{def of A_R})\ and utilized in the
decomposition~(\ref{case2}) (see Appendix). Thus, the solution always falls in
the invariant set $A$, indicating the fact that $x_{1}(t)=x_{2}(t)$\ and
$p_{1}(t)=p_{2}(t)$ for all $t$. This initial condition, however, seems too
special and very severe. Thus, we also treated another condition such that
$x_{1}(0)=x_{2}(0)=0,$ $p_{1}(0)=1,$ $p_{2}(0)=2$ (with $\zeta(0)=(0,0)$,
which was the same in all the cases below). Nevertheless, this initial
condition obeys a symmetry $S=\frac{1}{5}\left[
\begin{array}
[c]{cc}%
-3 & 4\\
4 & 3
\end{array}
\right]  \in O(2)$, so that this solution is also confined in the invariant
space, $\omega(t)\in A$\ for all $t$. Although it does not hold that
$x_{1}(t)=x_{2}(t)$\ and $p_{1}(t)=p_{2}(t)$ for all $t$,\ this confinement
severely effects the motion, and the trajectories of $x,p$\ and their
distributions exhibited similar nonergodic behavior as that in Figure 1 (not shown).

In the second case, we changed the initial condition only and set as
$x_{1}(0)=1,x_{2}(0)=0,$ $p_{1}(0)=0,$ $p_{2}(0)=0.01$, which does not have
any symmetry in $O(2)_{\times}$.\ The initial value\ $\omega_{0}$ is in
$A^{[+]}$\ because $\gamma_{12}(\omega_{0})=\frac{1}{2}(x_{1}(0)p_{2}%
(0)-p_{1}(0)x_{2}(0))>0$ (see Eq.~(\ref{def of gamma})\ and
Proposition~\ref{nonergodicity for nHO}. Thus the solution is not confined in
the \textquotedblleft small\textquotedblright\ subspace $A$\ but confined in
$A^{[+]}$, which is \textquotedblleft large\textquotedblright. However, as
shown in Figure~2, the distributions are far from the theoretical Gaussian
distributions and the trajectories are biased, suggesting a certain structure.
The third case we studied is the case with $x_{1}(0)=1,x_{2}(0)=0,$
$p_{1}(0)=0,$ $p_{2}(0)=-2$, which also has no symmetry in $O(2)_{\times}%
$\ and $\omega_{0}\in A^{[-]}$\ (because $\gamma_{12}(\omega_{0})<0$). This
yielded a relatively good results for trajectories and distributions (not
shown). Although a hall (which is smaller compared with that in the above
cases) was observed in the $(x,p)$\ trajectories and unignorable errors were
admitted in the distributions, it might be sufficient in practical
simulations. However, a clear numerical evidence for nonergodicity is a
definiteness of the signature of $\gamma_{12}(\omega(t))$. It should be a null
occurrence that $\gamma_{12}(\omega(t))=0$\ for all $t$, if the flow is
ergodic. Furthermore, there should not be the case where $\gamma_{12}%
(\omega(t))>0$\ for all $t$\ or $\gamma_{12}(\omega(t))<0$\ for all $t$.
Otherwise, it breaks the ergodicity and contradicts the BG distribution. In
fact, its average should be zero under the BG\ distribution: $\bar{\gamma
}_{ij}=\left\langle \gamma_{ij}\right\rangle =\left\langle \gamma
_{ij}\right\rangle _{\text{BG}}=0$ if the flow is ergodic with respect to
$\exp\left[  -\beta U(x)+K(p)\right]  \rho_{\text{Z}}\left(  \zeta\right)
d\omega$ for any smooth, positive, integrable $\rho_{\text{Z}}$\ (as long as
$\int_{D}x_{k}\exp\left[  -\beta U(x)\right]  dx$\ are finite for $k=i$ and
$j$). Figure 3 shows $\gamma_{12}(\omega(t))$\ for the three cases above. The
first case (Fig. 3a) corresponds to the null case $\gamma_{12}(\omega
(t))=0$\ for all $t$, and the second (Fig. 3b) and third (Fig. 3c)\ cases
correspond to $\gamma_{12}(\omega(t))>0$\ and $\gamma_{12}(\omega(t))<0$\ for
all $t$, respectively. These results show that the conventional method sampled
the phase space in a nonergodic manner. Note also that the magnitude of
$\gamma_{12}(\omega(t))$\ in the third case (Fig. 3c) is larger than the
second case (Fig. 3b). This result may be the reason why the third case shows
relatively good sampling; namely, trajectories staying near the invariant set
$A$\ shows bad sampling, whereas trajectories that can be away from
$A$\ relatively show good (but not exact) sampling. These staying features
near $A$\ may suggest a kind of stability of the invariant set $A$.

We tested the splitting NH EOM~(\ref{massive NH EOM}), currently provided
scheme, using the same harmonic oscillator system as above. A first example is
the case with $n=2$, where the initial value is the same as the most stiff
case used above, viz., $x(0)=(0,0),$ $p(0)=(1,1),$ $\zeta(0)=(0,0)$. We set
$\mathbf{Q}^{-1}$\ in the manner stated in Section \ref{We determine Q}, where
$\delta_{1}=0$, $\delta_{2}=0.2$, $\lambda=10$, and $\theta_{2,1}=0.5$ were
used. The scatter plots of all variables $x_{1}$, $x_{2}$, $p_{1}$, $p_{2}$,
$\zeta_{1}$, and $\zeta_{2}$ are shown in Figure 4. They are well sampled in
the phase space. The distributions agreed the theoretical distribution, and
the errors were sufficiently small, for which we have used variable $y$,
instead of $\zeta$, as indicated in Eq.~(\ref{distribution of y}). We also
observed that $\gamma_{12}(\omega(t))$\ does not indicate the
positive/negative definiteness, as in the conventional method, and rapidly
converged to the theoretical value $0$. We had similar results for other
initial conditions. Next, we show the results for other setting of
$\mathbf{Q}^{-1}$, where $\delta_{2}=0.8$ and $\theta_{2,1}=0.8$, while the
other conditions are the same as above. This is a setting where ellipsoid
distributed feature for $\zeta$\ is emphasized. We observe in Figure~5 that
$\zeta$ were distributed ellipsoidally\ around the origin and sampled
correctively, as indicated in the theoretical contours and the distributions
for $y$.

The next example for the splitting NH EOM~(\ref{massive NH EOM}) is the case
with $n=3$. In the procedures for setting $\mathbf{Q}^{-1}$, we put
$\theta_{3,1}=\theta_{3,2}=\theta_{2,1}=0.5$, $\delta_{1}=-0.2$, $\delta
_{2}=0$, $\delta_{3}=0.2$, and $\lambda=10$, and utilized
Eq.~(\ref{alt scheme for def of O}). Initial values were $x_{i}(0)=0,$
$p_{i}(0)=1,$ $\zeta_{i}(0)=0$ for all $i=1,2,3$ (same for the cases of
$n=2$). As shown in Figure 6, the scatter plots indicate the ergodic sampling,
and the distribution for each variable $x_{i}$, $p_{i}$, $\zeta_{i}$\ for
$i=1,2,3$ agrees with the theoretical distribution, respectively, as indicated
by the small errors. This also shows that the sampling were good even if
$\theta_{k,j}$\ were not set randomly. On the Basis of these results, we
conclude that the current method accurately corresponds to the ergodicity.

\section*{Appendix}

We say that a linear symmetry $S:\mathbb{R}^{n}\rightarrow\mathbb{R}^{n}%
$\ acts on ODE~(\ref{target EOM}) if it satisfies the following:

\begin{definition}
\label{condi of S} $S\in\mathrm{End}\mathbb{R}^{n}$ preserves the
functions\ $\lambda$\ and $\Lambda$\ and the domain $D$\ such that
$\lambda\left(  S(x),S(p),\zeta\right)  =\lambda\left(  x,p,\zeta\right)  $
and $\Lambda\left(  S(x),S(p),\zeta\right)  =\Lambda\left(  x,p,\zeta\right)
$ hold for all $\left(  x,p,\zeta\right)  \in\Omega$ and $S(D)\subset D$.
Commutativities also hold: $\mathbf{M}^{-1}\circ S=S\circ\mathbf{M}^{-1}$\ and
$F\circ S=S\circ F$.\ 
\end{definition}

We denote by $\mathcal{S}$\ the set of all $S\in\mathrm{End}\mathbb{R}^{n}%
$\ that acts on ODE~(\ref{target EOM}).

\begin{lemma}
For any $S\in\mathcal{S}$, we have $S(x(t))=x(t)$ and $S(p(t))=p(t)$\ for all
$t$ in an interval $J\subset\mathbb{R}$, if $\varphi:J\rightarrow
\Omega,t\mapsto\left(  x(t),p(t),\zeta(t)\right)  $ is a solution of
ODE~(\ref{target EOM})\ with an initial condition satisfying $S(x(0))=x(0)$
and $S(p(0))=p(0)$.
\end{lemma}

\begin{proof}
It follows from Definition~\ref{condi of S} that $\hat{\varphi}%
:\mathbb{R\supset}J\rightarrow\Omega,t\overset{\mathrm{d}}{\mapsto}\left(
S(x(t)),S(p(t)),\zeta(t)\right)  $ also becomes a solution of the $C^{1}$\ ODE
and $\hat{\varphi}(0)=\varphi(0)$ holds. Thus, the uniqueness of the initial
value problem ensures $\hat{\varphi}=\varphi$, so that $S(x(t))=x(t)$ and
$S(p(t))=p(t)$\ for all $t\in J$.
\end{proof}

We thus have for every $S\in\mathcal{S}$\ an invariant set,
\[
\Omega_{S}:=\Gamma_{s}^{D}\times\Gamma_{s}\times\mathbb{R}^{m},
\]
with%
\[
\Gamma_{s}\equiv\{p\in\mathbb{R}^{n}|\,S(p)=p\}
\]
and $\Gamma_{s}^{D}\equiv\Gamma_{s}\cap D\equiv\{x\in D$ $|$ $S(x)=x\}$,
indicating that the symmetry $(S(x),S(p))=(x,p)$\ is kept in the dynamics or
compatible with the ODE. To show the nonergodic
condition~(\ref{non ergode condi}), we take an approach that is to find \ an
invariant set $A$ whether it meets condition~(\ref{non ergode condi}) itself
or it separates the total phase space into three invariant sets,%
\begin{equation}
\Omega=A\sqcup A^{[+]}\sqcup A^{[-]}, \label{case2}%
\end{equation}
wherein $A^{[+]}$\ meets condition~(\ref{non ergode condi}). For this,
$A$\ should be "large" (for otherwise situation in choosing $A=\Omega_{S}$,
there is the extremely small case $\Gamma_{s}=\emptyset$ or a case of a low
dimensional subspace). Our target for $A$ is thus an invariant set that are
summed up these $\Omega_{S}$ in a certain manner:
\begin{equation}
A_{\mathcal{R}}\equiv\bigcup_{S\in\mathcal{R}}\Omega_{S}=\bigcup
_{S\in\mathcal{R}}(\Gamma_{s}^{D}\times\Gamma_{s})\times\mathbb{R}^{m},
\label{def of A_R}%
\end{equation}
where $\mathcal{R}$\ is a certain subset of $\mathcal{S}$\ such that it is
sufficiently large but not too large. For the latter condition, for example,
we should remove the case where $S$\ is the identity $\mathrm{id}%
_{\mathbb{R}^{n}}$, otherwise $A_{\mathcal{S}}$ becomes "too large"
($S=\mathrm{id}_{\mathbb{R}^{n}}$ provides $\Omega_{\mathrm{id}_{\mathbb{R}%
^{n}}}=\Omega$ and so yields $\Omega\backslash A_{\mathcal{R}}=\emptyset$,
which does not contribute to the nonergodic condition~(\ref{non ergode condi})
for $A\equiv A_{\mathcal{R}}$). \ 

We will show that~(\ref{case2}) holds\ with $A^{[\pm]}\equiv \Upsilon_{ij}%
^{\pm}$ if $A\equiv A_{\mathcal{R}}$, in a special case of the isotropic
$n$HO. Here $\Upsilon_{ij}^{\pm}$ are defined in Lemma~\ref{L1}, and this fact
can explain the route why $\Upsilon_{ij}^{\pm}$ arise. That is, they arise as
a complementary set to a sum, in a certain range $\mathcal{R}$, of the
invariant set $\Omega_{S}$\ based on the symmetry $S$\ that acts on the ODE.
To show the issue, we restrict the condition such that the dependence of
$x,p$\ in\ the functions $\lambda$\ and $\Lambda$\ is only through the
potential and kinetic energies; viz.,

\begin{condition}
\label{condi of A,B}There exist $C^{1}$\ functions $\tilde{\lambda}%
,\tilde{\Lambda}:\mathbb{R}\times\mathbb{R}\times\mathbb{R}^{m}\rightarrow
\mathbb{R}\ $such that$\ \lambda\left(  x,p,\zeta\right)  =\tilde{\lambda
}\left(  U(x),K(p),\zeta\right)  $ and $\Lambda\left(  x,p,\zeta\right)
=\tilde{\Lambda}\left(  U(x),K(p),\zeta\right)  $ for all $\left(
x,p,\zeta\right)  \in\Omega$
\end{condition}

This is not a special condition and are satisfied by Examples~\ref{NH}%
--\ref{GGMT}. For the harmonic oscillator system described by
Eq.~(\ref{isotropic condi}) under condition~\ref{condi of A,B}, orthogonal
transforms of $\mathbb{R}^{n}$ actually act on ODE~(\ref{target EOM}):

\begin{lemma}
$\mathcal{S\supset}O(n)\equiv\{S\in\mathrm{End}\mathbb{R}^{n}$ $|$
$^{\mathrm{T}}SS=\mathrm{id}_{\mathbb{R}^{n}}\}$ holds for the harmonic
oscillator system.
\end{lemma}

\begin{proof}
Take any $S\in O(n)$. Since $U(S(x))=U(x)$\ and $K(S(p))=K(p)$ hold for any
$(x,p)$, and since $\mathbf{M}^{-1}$ and $F$\ become diagonal, the conditions
in definition~\ref{condi of S} are valid, indicating $O(n)\subset\mathcal{S}$.
\end{proof}

We then put
\[
\mathcal{R}\equiv O(n)\backslash\{\mathrm{id}_{\mathbb{R}^{n}}\}=:O(n)_{\times
},
\]
viz., we sum up $\Omega_{S}$\ to make $A_{\mathcal{R}}$\ for all $S$\ that is
a non-identical orthogonal transform of $\mathbb{R}^{n}$. We also restrict our
consideration for $n=2$\ for ease, wherein the discussion would be extended to
a larger $n$. The following proposition shows how $\Upsilon_{12}^{\pm}%
$\ arises from $\mathcal{R}$.

\begin{proposition}
\label{decomposition}It holds that $A_{\mathcal{R}}=\Upsilon_{12}^{0}$, and
the decomposition~(\ref{case2}) holds with $A=A_{\mathcal{R}}$ and $A^{[\pm
]}=\Upsilon_{12}^{\pm}$.
\end{proposition}

\begin{proof}
By using the fact that $O(2)$\ is bijectively parametrized by $S^{1}$\ and
signature,
\[
S^{1}\times\{\pm1\}\rightarrow O(2),\text{ }(\theta,\sigma)\mapsto\left[
\begin{array}
[c]{cc}%
\cos\theta & -\sigma\sin\theta\\
\sin\theta & \sigma\cos\theta
\end{array}
\right]  =:s_{\theta}^{\sigma},
\]
and by solving an eigenvalue problem $s_{\theta}^{\sigma}(p)=p$, we see
\begin{align*}
\Gamma_{s_{\theta}^{+1}}  &  =\left\{
\begin{array}
[c]{cc}%
\mathbb{R}^{2} & \text{if }\theta=0\\
\{0,0\} & \text{otherwise}%
\end{array}
\right\}  ,\\
\Gamma_{s_{\theta}^{-1}}  &  =\left\{
\begin{array}
[c]{lc}%
\mathbb{R}\times\{0\} & \text{if }\theta=0\\
\{0\}\times\mathbb{R} & \text{if }\theta=\pi\\
\left\{  (p_{1},p_{2})\in\mathbb{R}^{2}\text{ }|\text{ }p_{2}=\frac
{1-\cos\theta}{\sin\theta}p_{1}\right\}  & \text{otherwise}%
\end{array}
\right\}  .
\end{align*}
Note that $\Gamma_{s_{\theta}^{-1}}$\ is a line through origin of
$\mathbb{R}^{2}$\ with a gradient $k=k_{\theta}$\ for $\theta\in
S^{1}\backslash\{0,\pi\}$, where $k_{\theta}$\ can take any value in
$\mathbb{R}_{\times}$, and that $\Gamma_{s_{0}^{-1}}=\mathbb{R}\times
\{0\}$\ and $\Gamma_{s_{\pi}^{-1}}=\{0\}\times\mathbb{R}$\ are also lines with
gradients\ $0$\ and $\infty$, respectively.\ Thus $\Gamma_{s_{\theta}^{-1}%
}=L_{k}$, a line through origin of $\mathbb{R}^{2}$\ with a gradient
$k\in(\mathbb{-\infty},\mathbb{\infty]}$. Applying $\mathcal{R=}O(2)_{\times
}=\{s_{\theta}^{+1}$ $|$ $\theta\in S_{\times}^{1}\}\cup\{s_{\theta}^{-1}$ $|$
$\theta\in S^{1}\}$, we hence get%
\begin{align*}
&  \bigcup_{S\in\mathcal{R}}(\Gamma_{s}^{D}\times\Gamma_{s})\\
&  =\bigcup_{\theta\in S_{\times}^{1}}(\Gamma_{s_{\theta}^{+1}}\times
\Gamma_{s_{\theta}^{+1}})\cup\bigcup_{\theta\in S^{1}}(\Gamma_{s_{\theta}%
^{-1}}\times\Gamma_{s_{\theta}^{-1}})\\
&  =\bigcup_{\theta\in S^{1}}(\Gamma_{s_{\theta}^{-1}}\times\Gamma_{s_{\theta
}^{-1}})\\
&  =\bigcup_{k\in(\mathbb{-\infty},\mathbb{\infty]}}(L_{k}\times L_{k})\\
&  =\{(x_{1},kx_{1},p_{1},kp_{1})|\text{ }x_{1},p_{1},k\in\mathbb{R}\}\\
&  \cup\text{ }\{0\}\times\mathbb{R}\times\{0\}\times\mathbb{R}\\
&  =\{(x_{1},x_{2},p_{1},p_{2})\in\mathbb{R}^{4}\text{ }|\text{ }x_{1}%
p_{2}-x_{2}p_{1}=0\}.
\end{align*}
Thus $A_{\mathcal{R}}=\bigcup_{S\in\mathcal{R}}(\Gamma_{s}^{D}\times\Gamma
_{s})\times\mathbb{R}^{m}=\{(x,p)\in\mathbb{R}^{4}|$ $x_{1}p_{2}-x_{2}%
p_{1}=0\}\times\mathbb{R}^{m}=\{\omega\in\Omega$ $|$ $\gamma_{12}%
(\omega)=0\}=\Upsilon_{12}^{0}$. Therefore, decomposition~(\ref{case2}) holds
as $\Omega=\Upsilon_{12}^{0}\sqcup \Upsilon_{12}^{+}\sqcup \Upsilon_{12}^{-}$.
\end{proof}

Note that the explicit form of $\Gamma_{s}$\ in the proof directly indicates
nontrivial examples to explain that $\mathcal{R}$\ should be sufficiently
large. For example, if we take $\mathcal{R}$\ as a one point set, then
$A_{\mathcal{R}}=\bigcup_{S\in\mathcal{R}}(\Gamma_{s}^{D}\times\Gamma
_{s})\times\mathbb{R}^{m}$ does not separate the phase space into the three
spaces as described in~(\ref{case2}): if we set $\mathcal{R}=\{s_{\theta}%
^{+1}\}$ with $\theta\neq0$ or $\mathcal{R}=\{s_{\theta}^{-1}\}$, then
$\bigcup_{S\in\mathcal{R}}(\Gamma_{s}^{D}\times\Gamma_{s})=\Gamma_{s_{\theta
}^{+1}}\times\Gamma_{s_{\theta}^{+1}}=\{0,0,0,0\}\subset\mathbb{R}^{4}$ or
$\bigcup_{S\in\mathcal{R}}(\Gamma_{s}^{D}\times\Gamma_{s})=\Gamma_{s_{\theta
}^{-1}}\times\Gamma_{s_{\theta}^{-1}}=$line$\times$line$\subset\mathbb{R}^{4}%
$, respectively, clearly induces no separation. So does for e.g. any finite
set $\mathcal{R}=\{s_{\theta_{1}}^{\pm1},\cdots,s_{\theta_{q}}^{\pm1}\}$.

\section*{ACKNOWLEDGMENTS}

This work was supported by a Grant-in-Aid for Scientific Research (C)
(17K05143 and 20K11854) from JSPS and the \textquotedblleft Development of
innovative drug discovery technologies for middle-sized
molecules\textquotedblright\ from Japan Agency for Medical Research and
development, AMED. We thank Profs. Haruki Nakamura and Akinori Kidera for
their continuous encouragement.

\bigskip

\bigskip\newpage

{\Huge Figure Captions}

\bigskip

Fig. 1. Simulation results obtained by a conventional thermostat method (NHC
with the chain length $2$) for the $2$HO using an initial condition
$x_{1}(0)=x_{2}(0)=0,$ $p_{1}(0)=p_{2}(0)=1$: (a) trajectory (scatter plot) of
$(x_{1},p_{1})$\ and marginal distributions\ for (b) $x_{1}$ and (c) $p_{1}$,
where corresponding theoretical distributions and the discrepancies are also
shown. The results for $x_{2}$ and $p_{2}$ are exactly the same as that for
$x_{1}$ and $p_{1}$\ (see text). \ 

\bigskip

Fig. 2. Simulation results obtained by the NHC for the $2$HO using
$x_{1}(0)=1,$ $x_{2}(0)=0,$ $p_{1}(0)=0,$ $p_{2}(0)=0.01$: (a) marginal
distributions\ for (a) $x_{1}$ and (b) $p_{1}$, and trajectory of $(\zeta
_{1},\zeta_{2})$.

\bigskip

Fig. 3. Trajectories of $\gamma_{12}(\omega)$\ obtained by the NHC for the
$2$HO using initial conditions of (a) that in Fig.~1, (b) that in Fig.~2, and
(c) $x_{1}(0)=1,x_{2}(0)=0,$ $p_{1}(0)=0,$ $p_{2}(0)=-2$.

\bigskip

Fig. 4. \ Simulation results obtained by a current thermostat method
(splitting NH) for the $2$HO using $x_{1}(0)=x_{2}(0)=0,$ $p_{1}%
(0)=p_{2}(0)=1$: (a) trajectories (scatter plot) of $(x_{1},p_{1})$\ and (b)
$(x_{2},p_{2})$, and the marginal distributions\ for (c) $x_{1}$ and $x_{2}%
$\ and (d) $p_{1}$ and $p_{2}$ (theoretical distributions and the
discrepancies are also shown); (e) trajectories of $(\zeta_{1},\zeta_{2}%
)$\ and (f) their marginal distributions\ represented in the principal
components $y_{1}$ and $y_{2}$; (g) trajectories of $\gamma_{12}(\omega)$\ and
its time average.

\bigskip

Fig. 5. Simulation results obtained by the splitting NH for the $2$HO, using
the different setting of the Nos\'{e} mass matrix $\mathbf{Q}$\ than that in
Fig.~4: (a) trajectory of $(\zeta_{1},\zeta_{2})$\ and the contours of the
theoretical distribution (rotated ellipsoids); the marginal distributions\ for
principal components (b) $y_{1}$ and (c) $y_{2}$.

\bigskip

Fig. 6. Simulation results obtained by the splitting NH for the $3$HO:
trajectories and distributions for (a)$\ x_{i}$, (b)$\ p_{i}$, and
(c)$\ y_{i}$ ($i=1,2,3$).
\end{document}